\documentclass[%
 reprint,
 amsmath,amssymb,
 aps, english,
]{revtex4-1}

\usepackage{graphicx}% Include figure files
\usepackage{dcolumn}% Align table columns on decimal point
\usepackage{bm}% bold math
\usepackage{refstyle}
\usepackage{amsmath}
\usepackage{amsthm}

\theoremstyle{plain}
\newtheorem{thm}{\protect\theoremname}

\makeatother

\usepackage{babel}
\providecommand{\theoremname}{Theorem}

\def\dbar{{\mathchar'26\mkern-12mu d}}

\begin{document}

\title{The Generalized Boltzmann Distribution is the Only Distribution in Which the Gibbs-Shannon Entropy Equals the Thermodynamic Entropy}

\author{Xiang Gao}
\affiliation{%
 Department of Chemistry, University of Florida{},
 Gainesville, Florida 32611
}%

\author{Emilio Gallicchio}
\affiliation{%
 Department of Chemistry, Brooklyn College of the City University of New York, and Doctoral Programs in Chemistry and Biochemistry, the Graduate Center of the City University of New York,
 Brooklyn, New York 11210
}%
\author{Adrian E. Roitberg}%
 \email{roitberg@ufl.edu}
\affiliation{%
 Department of Chemistry, University of Florida,
 Gainesville, Florida 32611
}%

\date{\today}% It is always \today, today,

\begin{abstract}
We show that the generalized Boltzmann distribution is the only distribution for which the Gibbs-Shannon entropy equals the thermodynamic entropy. This result means that the thermodynamic entropy and the Gibbs-Shannon entropy are not generally equal, but rather than the equality holds only in the special case where a system is in equilibrium with a reservoir.
\end{abstract}

\maketitle

%\tableofcontents

\section{Introduction}

There are two well known ways to derive the Boltzmann distribution:
the micro-canonical derivation, and the maximum entropy principle derivation.
Both are found in textbooks, such as \cite{kardar2007statistical, callen1998thermodynamics, chandler1987introduction, landau2013course, balescu1975equilibrium} for the micro-canonical derivation and \cite{chandler1988introduction, callen1998thermodynamics, tolman1979principles} for the maximum entropy principle derivation. These two derivations could be naturally extended to derive the generalized Boltzmann distributions for other ensembles such as grand canonical or  isothermal - isobaric ensemble. Beyond these two standard textbook derivations, the Boltzmann distribution can also be derived based on quantum dynamics\cite{tasaki1998quantum}.

Although modern statistical thermodynamics dates back to as early as Boltzmann\cite{cercignani1998ludwig} and Gibbs\cite{gibbs2014elementary}, new insights are still being obtained, such as the Jarzynski equality\cite{jarzynski1997nonequilibrium} and its quantum extension\cite{mukamel2003quantum}, as well as fluctuation theorems\cite{crooks1999entropy, evans1993probability, seifert2005entropy, esposito2010three}. There is also interest in entropy and its relationship with information, such as \cite{cerf1997negative, lieb2013entropy, parrondo2015thermodynamics}. In this paper, we study the basic question of the relationship between the generalized Boltzmann distribution, the thermodynamic entropy, and the Gibbs-Shannon entropy.

For context, we start by revisiting the two textbook derivations of the Boltzmann distribution.

\subsection{The micro-canonical derivation\label{subsec:micro}}

The micro-canonical derivation constructs an ensemble
with fixed total energy $E$ composed of the system of interest and
a reservoir. The fundamental postulate of statistical mechanics
states that the probability distribution of allowed micro-states is uniform:
\begin{equation}
p_k =\begin{cases}
\frac{1}{\Omega_E} & E_k=E\\
0 & \text{otherwise}
\end{cases}
\end{equation}
where $E_k$ is the energy of micro-state $k$ and
$\Omega_E$ is the number of micro-states with energy $E$.
The interaction between the system and reservoir is assumed to be
weak, in the sense that for a fixed micro-state $i$ of the system
with energy $E_i$, the reservoir is a micro-canonical ensemble
with energy $E_j=E-E_i$. The assumption of weak interaction also implies that is possible to enumerate the states of the system and reservoir independently. The probability of
micro-state $i$ of the system of interest is therefore
obtained by marginalizing over the allowed states $j$ of the reservoir
\begin{equation}
p_i =\sum_{j} \frac{1}{\Omega_E} = \frac{\Omega_{\text{r}(E - E_i)}}{\Omega_E}\propto\Omega_{\text{r}(E - E_i)} \label{eq:micrpi}
\end{equation}
where $\Omega_{\text{r}(E - E_i)}$ is the number of states of the reservoir
that satisfy $E_{j}=E-E_i$, which can be written as
\begin{equation}
\Omega_{\text{r}(E - E_i)}=\exp\left[S_{r(E - E_i)}/k_{B}\right]\label{eq:Omegar-Sr}
\end{equation}
using the definition of entropy for the micro-canonical ensemble. Since the energy of the system $E_s$ is a small fraction of the total energy, we can expand $S_r$ around $E$ as a power series of the system energy $E_s$ as
\begin{multline}
S_{r(E - E_i)} \simeq S_{r(E)}-\left(\frac{\partial S_{r}}{\partial E_{r}}\right)E_{i} \simeq S_{r}(E) -\frac{E_{i}}{T}\label{eq:Sr-expand}
\end{multline}
where $E_r = E - E_i$ is the energy of the reservoir and $T$ is the temperature of the reservoir. This equation becomes exact in the limit of an infinite reservoir at constant temperature T \cite{callen1998thermodynamics}.
Combining \eqref{micrpi}, \eqref{Omegar-Sr} and \eqref{Sr-expand}, we get $\Omega_{\text{r}(E - E_i)}\propto\exp\left(-\frac{E_i}{k_{B}T}\right)$,
or, equivalently, $p_i \propto\exp\left(-\frac{E_i}{k_{B}T}\right)$, which is the Boltzmann distribution.

It is worth mentioning that the assumption of weak interaction between the bath and the system can be relaxed by canonical typicality\cite{goldstein2006canonical, reimann2007typicality}.

\subsection{The maximum entropy principle derivation}

The maximum entropy principle\cite{jaynes1957information,jaynes1957information2}
derives the Boltzmann distribution by maximizing the Gibbs-Shannon
entropy\cite{cover2012elements} $H=-\sum_{i}p_{i}\log p_{i}$ under the constraints of $\left\langle E\right\rangle =\sum_{i}p_{i}E_{i}$
being a constant $E_0$ and of $\sum_{i}p_{i}=1$. The derivation is done using the
Lagrangian multiplier method which maximizes the target function
\begin{equation}
\mathcal{L}=-\sum_{i}p_{i}\log p_{i}-\beta\left(\sum_{i}p_{i}E_{i}-E_{0}\right)+\alpha\left(\sum_{i}p_{i}-1\right)
\end{equation}
 where $\alpha$ and $\beta$ are both Lagrangian multipliers. Zeroing
the derivative of $\mathcal{L}$ with respect to $p_{i}$ gives
\begin{equation}
0=\frac{\delta\mathcal{L}}{\delta p_{i}}=-1-\log p_{i}-\beta E_{i}+\alpha\Rightarrow p_{i}=C\cdot e^{-\beta E_{i}}
\end{equation}
This optimization process has been shown to be equivalent to the micro-canonical
derivation as shown in subsection \ref{subsec:micro} by applying the maximum entropy principle to the whole isolated
system containing the system of interest and reservoir, in a two step fashion\cite{PhysRevE.86.041126}.

\subsection{Connecting the canonical ensemble with thermodynamics}

After obtaining the probability distribution of the ensemble, we need to establish the connection between thermodynamic variables and ensemble quantities. For the canonical ensemble, the thermodynamic temperature $T$, volume $V$, and particle number $N$ are naturally mapped to the $T, V, N$ parameter of the ensemble.
Since the energy of the ensemble is not a parameter but instead a random variable, the mapping of the thermodynamic internal energy $U$ is not as obvious as it is for $T, V, N$.

One then introduce a new postulate that equates the thermodynamic energy $U$ with the ensemble average of the random variable $E$ of system energy:
\begin{multline}
U=\left\langle E\right\rangle=\sum_i p_i E_i = \frac{\sum_i E_i e^{-\beta E_i}}{\sum_i e^{-\beta E_i}}=-\frac{\partial \log Z}{\partial\beta} \label{eq:u-is-avg-e}
\end{multline}
where $Z=\sum_i e^{-\beta E_i}$ is the partition function and $\beta=\frac{1}{k_B T}$.
Since $d\left(\beta F\right)=U d\beta-\beta p dV$, i.e. $\beta F$ has natural variables $\beta$ and $V$, and $U=\frac{\partial \left(\beta F\right) }{\partial \beta}$. 
Comparing this relation with \eqref{u-is-avg-e} immediately gives the equation for Helmholtz free energy $F=-k_B T \log Z$. From the definition of Helmholtz free energy $F=U-TS$, we immediately obtain that the thermodynamic entropy can then be computed as
\begin{equation}
S=\frac{U-F}{T}=-k_B \sum_i p_i \log p_i = k_B H
\end{equation}
where $H=-\sum_i p_i \log p_i$ is the Gibbs-Shannon entropy.

\subsection{Our contribution}

The above logic shows that
\[
\left.\begin{array}{c}
\left(\begin{array}{c}
\text{Thermodynamic}\\
\text{first law}
\end{array}\right)\\
+\\
\left(\begin{array}{c}
\text{The Boltzmann}\\
\text{distribution}
\end{array}\right)
\end{array}\right\} \Rightarrow\left(\begin{array}{c}
\text{Thermodynamic entropy}\\
\text{equals\footnote{Here ``equals'' means differs by a coefficient $k_B$}}\\
\text{Gibbs-Shannon entropy}
\end{array}\right)
\]
while we will now prove the opposite direction:
\[
\left.\begin{array}{c}
\left(\begin{array}{c}
\text{Thermodynamic}\\
\text{first law}
\end{array}\right)\\
+\\
\left(\begin{array}{c}
\text{Thermodynamic entropy}\\
\text{equals}\\
\text{Gibbs-Shannon entropy}
\end{array}\right)
\end{array}\right\} \Rightarrow\left(\begin{array}{c}
\text{The Boltzmann}\\
\text{distribution}
\end{array}\right)
\]
The contribution of this paper is therefore two fold: (1) shows that the generalized Boltzmann distribution is the only distribution in which the Gibbs-Shannon entropy equals the thermodynamic entropy, and (2) presents a new way of deriving the generalized Boltzmann distribution. 

\section{Main result}

We consider a thermodynamic system with $m+n$ generalized forces and coordinates\cite{kubo1968thermodynamics}. The thermodynamic first law can be written as
\begin{equation}
dU=TdS+\sum_{\eta=1}^{n}X_{\eta}d\chi_{\eta}+\sum_{\eta=1}^{m}Y_{\eta}dy_{\eta}\label{eq:first-law}
\end{equation}
We use $U, \chi_1, \ldots \chi_n$ to denote thermodynamic quantities and $E, x_1, \ldots, x_n$ to denote random variables.
We also intentionally choose different letters for generalized forces ($X_i$ and $Y_i$) and displacements ($\chi_i$ and $y_i$) to emphasize that the ensemble we are going to study is parameterized by $T$,  generalized forces $X_1,\ldots,X_n$, and generalized displacements $y_1\ldots,y_m$; that is, $E, x_1\ldots x_n$ are random variables, while $T, X_1, \ldots, X_n, y_1, \ldots, y_n$ are parameters.
In other word, we are studying a $(T, X_1, \ldots, X_n, y_1, \ldots, y_n)-$ ensemble.

The above notation is a general framework for all thermodynamic systems. For example, for an $(\mu, V, T)-$ ensemble (one component grand canonical ensemble), we have $m=n=1$, $y_1=V$, $Y_1=-p$, $\chi_1=N$, $X_1=\mu$.

The main result we present in this paper is the following theorem:
\begin{thm}
Consider thermodynamic systems whose first law reads like \eqref{first-law}.
Any $\left(T,X_{1},\ldots,X_{n},y_{1},\ldots,y_{m}\right)$ ensemble
that has:
\begin{enumerate}
\item Probability density proportional to some function of the ensemble parameters $T,X_{1},\ldots,X_{n},y_{1},\ldots,y_{m}$ and random variables $E,x_{1},x_{2},\cdots,x_{n}$.
\item $\left\langle E\right\rangle =U$, $\left\langle x_{\eta}\right\rangle =\chi_{\eta}$
for all $\eta=1,\ldots,n$
\item The thermodynamic entropy $S$ and the Gibbs-Shannon entropy $H$ have a relationship
$S=k_{B}H$
\end{enumerate}
the probability density function of micro-state $\vec{\omega}$ must have the form
\begin{equation}
Pr\left(\vec{\omega}\right)\propto\exp\left[\sum_{\eta=1}^{n}\frac{X_{\eta}x_{\eta}^{\left(\vec{\omega}\right)}}{k_{B}T}-\frac{E^{\left(\vec{\omega}\right)}}{k_{B}T}\right]
\end{equation}
where $E^{\left(\vec{\omega}\right)},x_{1}^{\left(\vec{\omega}\right)},\ldots,x_{n}^{\left(\vec{\omega}\right)}$
are the value of random variables for the state $\vec{\omega}$.
\end{thm}

\begin{proof}
Let's begin by assuming the probability density of micro-state $\vec{\omega}$
is expressed as
\begin{equation}
\Pr\left(\vec{\omega}\right)\propto f\left(E^{\left(\vec{\omega}\right)},x_{1}^{\left(\vec{\omega}\right)},\ldots,x_{n}^{\left(\vec{\omega}\right)};T,X_{1},\ldots,y_{1},\ldots\right)\label{eq:propto}
\end{equation}
the normalization constant is

\begin{equation}
Z=\sum_{\vec{\omega}}f_{\vec{\omega}}
\end{equation}
where the sum is over all microscopic states and $f_{\vec{\omega}}$
is short for $f\left(E^{\left(\vec{\omega}\right)},x_{1}^{\left(\vec{\omega}\right)},\ldots,x_{n}^{\left(\vec{\omega}\right)};T,X_{1},\ldots,y_{1},\ldots\right)$. In the above formula,
the temperature $T$ appears as a parameter, and micro-states $\vec{\omega}$
and its macro-states $E^{\left(\vec{\omega}\right)},x_{1}^{\left(\vec{\omega}\right)},\ldots,x_{n}^{\left(\vec{\omega}\right)}$
do not depend on $T$. Note that if $S=k_B H$, then the thermodynamic entropy
\begin{multline}
S=-k_{B}\sum_{\vec{\omega}}\frac{f_{\vec{\omega}}}{Z}\log\left(\frac{f_{\vec{\omega}}}{Z}\right) \\
=-k_{B}\left\{ \frac{\sum_{\vec{\omega}}f_{\vec{\omega}}\log f_{\vec{\omega}}}{Z}-\log Z\right\} 
\end{multline}
If we perturb the temperature $T$ by $dT$ and keep the other parameters ($X_{1},\ldots,X_{n},y_{1},\ldots,y_{m}$) fixed,
then the change of entropy
\begin{equation}
dS=\left(\frac{\partial S}{\partial T}\right)_{X_{1},\ldots,X_{n},y_{1},\ldots,y_{m}}dT\label{eq:dSdT}
\end{equation}
where
\begin{multline}
\left(\frac{\partial S}{\partial T}\right)_{X_{1},\ldots,X_{n},y_{1},\ldots,y_{m}} \\
=-k_{B}\left\{ \sum_{\vec{\omega}} \frac{Z\cdot\frac{\partial f_{\vec{\omega}}}{\partial T}\cdot\log f_{\vec{\omega}}-\frac{\partial Z}{\partial T}\cdot f_{\vec{\omega}}\log f_{\vec{\omega}}}{Z^{2}}\right\} \\
=-k_{B}\sum_{\vec{\omega}}\frac{\partial}{\partial T}\left(\frac{f_{\vec{\omega}}}{Z}\right)\cdot\log\left(f_{\vec{\omega}}\right)\label{eq:pSpT}
\end{multline}
Also
\begin{equation}
\left\langle E\right\rangle =\sum_{\vec{\omega}}\frac{f_{\vec{\omega}}}{Z}\cdot E^{\left(\vec{\omega}\right)}
\end{equation}
and therefore
\begin{multline}
d\left\langle E\right\rangle =\left(\frac{\partial\left\langle E\right\rangle }{\partial T}\right)_{X_{1},\ldots,X_{n},y_{1},\ldots,y_{m}}dT \\
=\sum_{\vec{\omega}}\frac{\partial}{\partial T}\left(\frac{f_{\vec{\omega}}}{Z}\right)\cdot E^{\left(\vec{\omega}\right)}\cdot dT\label{eq:dEdT}
\end{multline}
The same argument applies for all $\left\langle x_{\eta}\right\rangle $,
and we therefore have
\begin{multline}
d\left\langle x_{\eta}\right\rangle =\left(\frac{\partial\left\langle x_{\eta}\right\rangle }{\partial T}\right)_{X_{1},\ldots,X_{n},y_{1},\ldots,y_{m}}dT \\
=\sum_{\vec{\omega}}\frac{\partial}{\partial T}\left(\frac{f_{\vec{\omega}}}{Z}\right)\cdot x_{\eta}^{\left(\vec{\omega}\right)}\cdot dT\label{eq:dxdT}
\end{multline}

By substituting equations in condition 2 of the theorem into \eqref{first-law}, the first law of thermodynamics can be rewritten as
\begin{equation}
-\frac{dS}{k_{B}}+\frac{d\left\langle E\right\rangle }{k_{B}T}-\frac{1}{k_{B}T}\sum_{\eta=1}^{n}X_{\eta}d\left\langle x_{\eta}\right\rangle -\frac{1}{k_{B}T}\sum_{\eta=1}^{m}Y_{\eta}dy_{\eta}=0
\end{equation}
substituting \eqref{dSdT}, \eqref{pSpT}, \eqref{dEdT}, \eqref{dxdT}
and $dy_{\eta}=0$\footnote{Because we are perturbing $T$ while keeping $y_\eta$ constant}, we have
\begin{equation}
\sum_{\vec{\omega}}\frac{\partial}{\partial T}\left(\frac{f_{\vec{\omega}}}{Z}\right)\cdot\left[\log\left(f_{\vec{\omega}}\right)+\frac{E^{\left(\vec{\omega}\right)}}{k_{B}T}-\sum_{\eta=1}^{n}\frac{X_{\eta}x_{\eta}^{\left(\vec{\omega}\right)}}{k_{B}T}\right]dT=0
\end{equation}

Physics laws should be universal, i.e. the above equation must hold
for arbitrary system, the only way for this to happen is
\begin{equation}
\log\left(f_{\vec{\omega}}\right)+\frac{E^{\left(\vec{\omega}\right)}}{k_{B}T}-\sum_{\eta=1}^{n}\frac{X_{\eta}x_{\eta}^{\left(\vec{\omega}\right)}}{k_{B}T}=0
\end{equation}
that is
\begin{equation}
f_{\vec{\omega}}=\exp\left[\sum_{\eta=1}^{n}\frac{X_{\eta}x_{\eta}^{\left(\vec{\omega}\right)}}{k_{B}T}-\frac{E^{\left(\vec{\omega}\right)}}{k_{B}T}\right]
\end{equation}
We therefore obtained the generalized Boltzmann distribution.
\end{proof}

\section{Examples}

The formulation of our main result is quite general. In this section, we will show how commonly used ensembles fit into our general framework.

The single component canonical ensemble, i.e. the $(N, V, T)$ ensemble, has $y_1=N$, $y_2=V$, $Y_1=\mu$, $Y_2=-p$. There are no $\chi$, $X$ or $x$. The first law reads $dU=TdS-pdV+\mu dN$ and the distribution is $\Pr\left(\vec{\omega}\right)\propto\exp{\left(-\frac{E^{\left(\vec{\omega}\right)}}{k_{B}T}\right)}$.

The two component grand canonical ensemble, i.e. the $(\mu_1, \mu_2, V, T)$ ensemble, has $y_1=V$, $Y_1=-p$, $\chi_1=N_1$, $\chi_2=N_2$, $X_1=\mu_1$, $X_2=\mu_2$, $x_1^{\left(\vec{\omega}\right)}=N_1^{\left(\vec{\omega}\right)}$, and $x_2^{\left(\vec{\omega}\right)}=N_2^{\left(\vec{\omega}\right)}$. The first law reads $dU=TdS-pdV+\mu_1 dN_1+\mu_2 dN_2$ and the distribution is $\Pr\left(\vec{\omega}\right)\propto\exp{\left(\frac{\mu_1 N_1^{\left(\vec{\omega}\right)}+\mu_2 N_2^{\left(\vec{\omega}\right)}-E^{\left(\vec{\omega}\right)}}{k_{B}T}\right)}$.

The single component isothermal--isobaric ensemble, i.e. the $(N, p, T)$ ensemble, has $y_1=N$, $Y_1=\mu$, $\chi_1=V$, $X_1=-p$, and $x_1^{\left(\vec{\omega}\right)}=V^{\left(\vec{\omega}\right)}$. The first law reads $dU=TdS-pdV+\mu dN$ and the distribution is
$\Pr\left(\vec{\omega}\right)\propto\exp{\left(-\frac{pV^{\left(\vec{\omega}\right)}+E^{\left(\vec{\omega}\right)}}{k_{B}T}\right)}$.

\section{Discussion}

We showed that the ensemble for a non-isolated thermodynamic system at equilibrium described by temperature, parameters $y_i$ and expectation values $\left\langle x_i\right\rangle$, and whose thermodynamic entropy is given by the Gibbs-Shannon formula, is necessarily described by the generalized Boltzmann's distribution. It follows that the generalized Boltzmann distribution is the only distribution for which the Gibbs-Shannon entropy equals the thermodynamic entropy.

Unlike familiar thermodynamic quantities such as energy, volume, and pressure, the concept of thermodynamic entropy is harder to grasp. The entropy is nevertheless a well-defined state function of systems at equilibrium, stemming from the fact that $\oint \dbar Q/T = 0$ for any closed thermodynamic path along equilibrium states. Conversely, the thermodynamic entropy is defined only for states at equilibrium\cite{lieb2013entropy}. So, while the question of what is the form of the entropy for states out of equilibrium is fundamentally ill-posed, it is tempting to try to extend the concept of thermodynamic entropy to stable non-equilibrium systems\cite{parrondo2015thermodynamics}, such as systems at steady state. Our result implies that the entropy of such systems cannot be described by the Gibbs-Shannon formula unless their distribution of states follows the generalized Boltzmann's distribution.

Boltzmann famously defined the entropy of an isolated system at energy $E$ as $S(E)=-k_B \log {p(E)}$ \cite{cercignani1998ludwig}, $S(E)$ is the micro-canonical entropy and where $p(E)=1/\Omega$ is the probability of occupancy of any of the degenerate micro-states of the system. The identification of the thermodynamic entropy of a system in equilibrium with a reservoir with the ensemble average of the micro-canonical entropy ($S = -k_B \left\langle \log p \right\rangle = -k_B \sum_i p_i \log{p_i}$), as in the Gibbs-Shannon entropy formula\cite{gibbs2014elementary}, while known to be consistent with the Boltzmann's distribution, appears arbitrary. For example, it is worth considering whether there could be another distribution whose Gibbs-Shannon entropy can be identified with the thermodynamic entropy. Our result shows that within the very general thermodynamic assumptions (\eqref{first-law, propto}) this is not possible.

It is also worth mention that, for the special case of $(N,V,T)$ ensemble, instead of starting from the thermodynamic first law in \eqref{first-law}, similar results could also be obtained\cite{yuhan2019personal} starting from the expression of heat of quantum thermodynamics\cite{ma2017quantum, quan2005quantum, su2018heat, quan2007quantum}, i.e. $\dbar Q=\sum_i E_i dp_i$: Under the condition $S=k_B H$:
\begin{equation}
\sum_{i}E_{i}dp_{i}=\dbar Q=TdS=-k_{B}T\sum_{i}\log p_{i}dp_{i}
\end{equation}
which immediately gives $p_i\propto \exp{\left(-\frac{E_i}{k_B T}\right)}$.

\section{Acknowledgement}

The authors would like to thank Vin\'{i}cius Wilian D. Cruzeiro for proofreading of this paper, suggestions on the organization of this paper, and technical discussion on the condition of our main results. Discussions with Daniel Zuckerman, John Chodera, and Mike Gilson, have enriched this manuscript.

The authors would also like to thank Yuhan Ma for the discussion on the relationship of the main result of this paper with thermal equilibrium, as well as the connection of the proof of the main result with respect to quantum thermodynamics.

\bibliographystyle{unsrtnat}
\bibliography{main}

\end{document}